\documentclass[a4paper,11pt]{article}
\pdfoutput=1 

\usepackage{jheppub} 

\usepackage[utf8]{inputenc} 

\usepackage{amsmath}
\usepackage{amssymb}
\usepackage{amsthm}
\usepackage{mathrsfs}
\usepackage[all]{xy} 

\newcommand{\vol}{\mathrm{vol}}

\newcommand{\Ric}[1]{\operatorname{Ric}(#1)}

\newtheorem{theorem}{Theorem}

\newtheorem{remark}{Remark}

\newtheorem{lemma}{Lemma}

\newtheorem{question}{Question} 

\newcommand{\dd}{\mathrm{d}}
\newcommand{\ii}{\mathrm{i}}

\usepackage{xcolor}

\title{Harmonic band theory: rigidity of non-zero degree harmonic maps from 2-torus to complex projective space}

\author[a]{Yoshinori Hashimoto}
\author[b,c]{Bruno Mera}
\author[c]{Tomoki Ozawa}
\affiliation[a]{Department of Mathematics, Osaka Metropolitan University, 3-3-138, Sugimoto, Sumiyoshi-ku,
Osaka, 558-8585, Japan}
\affiliation[b]{Instituto de Telecomunica\c{c}\~oes and Department of Mathematics, Instituto Superior T\'ecnico, Universidade de Lisboa, 1049-001 Lisboa, Portugal} 	
\affiliation[c]{Advanced Institute for Materials Research (WPI-AIMR), Tohoku University, Sendai 980-8577, Japan} 
\emailAdd{yhashimoto@omu.ac.jp}
\emailAdd{bruno.mera@gmail.com}
\emailAdd{tomoki.ozawa.d8@tohoku.ac.jp}
\makeatletter
\gdef\@fpheader{}
\makeatother
\abstract{We prove the rigidity of isotropic harmonic maps from a 2-torus to a complex projective space, when they are constructed from holomorphic embeddings associated to complete linear systems. We also prove that this rigidity holds for any holomorphic embeddings without special hyperosculation points, with an extra assumption on the pullbacks of Fubini--Study symplectic forms. These results ensure the rigidity of towers of harmonic bands in condensed matter physics.}
\begin{document}
\maketitle
\section{Introduction}
\label{sec: Introduction}
In recent years, condensed matter physics has undergone a paradigm shift: the notion of a band insulator is no longer viewed merely as an energy–crystal momentum graph with a gap between valence and conduction bands. At zero temperature, the ground state is specified by the \emph{(occupied) Bloch vector bundle}: the Hermitian rank-$r$ vector bundle of valence bands, equipped with the Berry connection, over the Brillouin zone $\mathbb{T}^d$, which labels the unitary irreducible representations of the crystal’s translational symmetry. Mathematically, this bundle can be described by a smooth map $f: \mathbb{T}^d \to \mathrm{Gr}_r(\mathbb{C}^{N})$,
where $r$ is the number of valence bands, $N$ is the total number of bands, and the bundle itself is the pullback under $f$ of the tautological bundle over $\mathrm{Gr}_r(\mathbb{C}^N)$.
The first Chern class of its determinant line bundle, represented by the Berry curvature, is essentially the Hall conductivity by the celebrated result of Thouless, Kohmoto, Nightingale, den Nijs, and Wu~\cite{thouless:kohmoto:nightingale:den_Nijs:1982,kohmoto:1985,niu:thouless:wu:1985}.

Beyond topological invariants, geometric quantities such as the \emph{Berry curvature} and the \emph{quantum metric} — together called the \emph{quantum geometry} — have emerged as crucial in flat-band physics, influencing a wider range of phenomena beyond the non-interacting limit, ranging from fractional Chern insulators, to topological superconductivity, to non-linear optical responses, to name a few~\cite{torma:peotta:bernevig:2022,torma:2023,gao:nagaosa:ni:xu:2025}.

In two dimensions, $d=2$, the quantum geometry satisfies an inequality that is saturated by the lowest Landau level and, more generally, by what are now known as \emph{K\"ahler bands}~\cite{ozawa:mera:2021,mera:ozawa:2021,mera:ozawa:2021:engineering}. In the special case of a translation-invariant complex structure (see Sec.~\ref{subsec: Kaehler bands}), these are referred to as \emph{ideal bands}~\cite{ozawa:mera:2021,mera:ozawa:2021,mera:ozawa:2021:engineering,wang:cano:millis:liu:yang:2021,wang:liu:2022,ledwith:vishwanath:khalaf:2022,wang:klevtsov:liu:2023,estienne:regnault:valentin:2023,liu:mera:fujimoto:ozawa:2025}. When $f$ is an immersion, saturation is equivalent to $f$ being holomorphic, which is essentially the mathematical definition of a K\"ahler band. While ideal bands capture lowest Landau level physics, reproducing higher Landau level physics is of interest for realizing non-Abelian topological phases.

A natural generalization of holomorphic maps is given by \emph{harmonic maps}. We focus on maps $f: C \to \mathbb{P}^{N-1}$ from an elliptic curve $C$ (a $2$-torus with a translation-invariant complex structure), where we write $\mathbb{P}^{N-1}$ for the complex projective space $\mathbb{CP}^{N-1}$ throughout this paper. By the theorem of Eells and Wood~\cite{EW}, any such map can be obtained from a holomorphic one by taking derivatives and applying the Gram–Schmidt process. This construction is directly analogous to the action of the Landau-level creation operator, and coincides with it in the case of Landau levels. The higher Landau levels correspond to harmonic maps of non-zero degree that are not holomorphic. The corresponding bands are termed \emph{harmonic bands}\cite{paiva:wang:ozawa:mera:2025, onishi:avdoshkin:fu:2025}, and this notion coincides with the generalized Landau levels introduced in Ref.~\cite{liu:mera:fujimoto:ozawa:2025}.

In Ref.~\cite{liu:mera:fujimoto:ozawa:2025}, harmonic bands were numerically shown to stabilize non-Abelian fractional Chern insulator phases, thereby highlighting their similarity to conventional Landau levels beyond the lowest Landau level. From a physics perspective, classifying harmonic bands is an important problem, as it may reveal whether different harmonic bands can give rise to distinct exotic phenomena. 
Calabi's rigidity theorem~\cite{Cal} shows that ideal bands---the analogs of the lowest Landau level---are determined, up to projective isometry, by their quantum metric~\cite{wang:klevtsov:liu:2023,advoshkin:popov:2023,mera:ozawa:2024,liu:mera:fujimoto:ozawa:2025,paiva:wang:ozawa:mera:2025}. Here we prove an analogous rigidity for general harmonic bands (Theorem~\ref{thmrgd}).
\section{Physical background}
\label{sec: physical background}
We briefly review the tight-binding formalism in the simplest setting, the emergence of the Bloch Hamiltonian, the map $f:\mathbb{T}^d \to \mathbb{P}^{N-1}$, the notion of K\"ahler bands and also harmonic bands. The aim of this section is to provide physical motivation and context for our work. 

Consider a $d$-dimensional tight-binding model of a band insulator, where fermions hop between lattice sites in a translation-invariant way. In condensed matter, the lattice is $\mathbb{Z}^d$—a lattice in the mathematical sense of a finitely generated free $\mathbb{Z}$-module—or, for finite systems, the discrete torus
\begin{align}
X = \mathbb{Z}^d / (L_1\mathbb{Z} \times \cdots \times L_d\mathbb{Z}) \cong \mathbb{Z}_{L_1} \times \cdots \times \mathbb{Z}_{L_d},
\end{align}
with $L_i$ the system size in the $i$th direction, assumed to be large but finite.

The single-particle Hilbert space $\mathcal{H}$ consists of maps $\psi: \mathbb{Z}^d \to \mathbb{C}^N$ satisfying periodic boundary conditions
\begin{align}
\psi(x_1,\dots,x_i+L_i,\dots,x_d) =\psi(x_1,\dots,x_d), \quad i=1,\dots,d.
\end{align}
Here $\mathbb{C}^N$ accounts for internal degrees of freedom (e.g., spin or sublattice). 

Translation by $\gamma\in\mathbb{Z}^d$ acts as $(\gamma\cdot\psi)(x) = \psi(x+\gamma)$. A translation-invariant tight-binding Hamiltonian $H$ is a Hermitian operator commuting with this action, determined by a hopping function $h:\mathbb{Z}^d \to \mathrm{Mat}_{N\times N}(\mathbb{C})$:
\begin{align}
(H\psi)(x) = \sum_{\gamma\in\mathbb{Z}^d} h(\gamma)\,\psi(x+\gamma),
\end{align}
with $h(\gamma)$ of finite range (or rapid decay).

Fourier decomposition yields
\begin{align}
\mathcal{H} = \bigoplus_k \mathcal{H}(k),
\label{eq: direct sum decomposition of Hilbert space}
\end{align}
where $k=(k_1,\dots, k_d)$, $k_j = m_j/L_j$, $j=1,\dots, d$ and, $\mathcal{H}(k)$ is defined by $\psi\in \mathcal{H}(k)\iff \psi(x+\gamma) = e^{2\pi \ii k\cdot\gamma} \psi(x)$ for any $x,\gamma\in\mathbb{Z}^d$. Observe that this actually implies that $\psi(x)=e^{2\pi \ii k\cdot x}u_k$, for $x$-independent $u_k\in \mathbb{C}^N$. 
The $k$ occurring here are known as the crystal momenta or quasimomenta and they parametrize the unitary irreducible representations of $\mathbb{Z}^d$. The set of all crystal momenta $k$ is known as the \emph{Brillouin zone} and is naturally a torus $\mathbb{T}^d$.

In this basis, $H$ is block diagonal:
\begin{align}
H = \bigoplus_k h(k),
\end{align}
with $h(k)$ the discrete Fourier transform of the hopping function $\mathbb{Z}^d\ni \gamma\mapsto h(\gamma)\in \mathrm{Mat}_{N\times N}(\mathbb{C})$. For notational convenience, and with a slight abuse of notation, we use the same symbol for the hopping function and its Fourier transform. Under locality assumptions, $h(k)$ is a smooth, $\mathbb{Z}^d$-periodic map to Hermitian $N\times N$ matrices---the \emph{Bloch Hamiltonian}. Its spectrum as a function of $k$ gives the band structure.

The system is a band insulator if the Fermi projector
\begin{align}
P(k) =  \Theta\left(E_F I_N - h(k)\right),
\end{align}
is smooth and periodic in $k$. Here, $\Theta$ is the Heaviside stepfunction, $I_N$ denotes the $N\times N$ identity matrix and $E_F$ is the Fermi energy, which we can set to zero without loss of generality.

The map $f:\mathbb{T}^d \to \mathrm{Gr}_r(\mathbb{C}^N)$ mentioned in the introduction is defined by
\begin{align}
\mathbb{T}^d\ni [k]\longmapsto f([k]) = \mathrm{Im}\; P(k) \subset \mathbb{C}^N
\end{align}
The ground state is obtained by filling all states below the Fermi energy. This corresponds to the unique many-particle fermionic state determined by $P$, namely
\begin{align}
\Psi \in \bigwedge_{k} \left[\bigwedge^r \mathrm{Im}\, P(k)\right]
\subset \bigwedge_{k} \left(\bigwedge^r \mathbb{C}^N\right)
\subset \Lambda^* \mathcal{H},
\label{eq: ground state at boundary condition varphi}
\end{align}
where $\Lambda^* \mathcal{H}$ is the vector space underlying the exterior algebra of $\mathcal{H}$, i.e., the full many-body Hilbert space.

Here $E_{[k]} := \mathrm{Im}\; P(k)$ are the fibers of the rank $r$ occupied Bloch vector bundle $E \to \mathbb{T}^d$, the pullback under $f$ of the tautological vector bundle over the Grassmannian. The fibers of its determinant line bundle $\det(E) \to \mathbb{T}^d$ are $\det(E)_{[k]} := \bigwedge^r \mathrm{Im}\, P(k)$. The ground state is therefore described by the line bundle $\det(E) \to \mathbb{T}^d$. One might argue that the ground state does not depend on the smooth vector bundle structure, since only a finite set of fibers appears in its construction. In fact, the smooth structure is essential.

Indeed, in the thermodynamic limit $L_i \to \infty$ ($i=1,\dots,d$), the direct sum decomposition in Eq.~\eqref{eq: direct sum decomposition of Hilbert space} becomes a direct integral, realizing the Hilbert space isomorphism
\begin{align}
\ell^2(\mathbb{Z}^d; \mathbb{C}^N) \cong L^2(\mathbb{T}^d; \mathbb{C}^N).    
\end{align}
In this limit, the entire bundle $E$ appears through the ground-state projector, which is smooth in $k$—whereas before we only sampled this projector at finitely many points.

We are thus led to the key geometric object associated with the ground state:
\begin{align}
\iota \circ f: \mathbb{T}^d \to \mathbb{P}(\bigwedge^r \mathbb{C}^N), 
\end{align}
where $\iota:\mathrm{Gr}_r(\mathbb{C}^N) \hookrightarrow \mathbb{P}(\bigwedge^r \mathbb{C}^N)$ is the Pl\"ucker embedding. Observe that, indeed, $\iota\circ f([k])=\bigwedge^r E_{[k]}= \bigwedge^r \mathrm{Im}\; P(k)$ are precisely the factors that appear in $\Psi$, see Eq.~\eqref{eq: ground state at boundary condition varphi}. Since $\mathbb{P}(\bigwedge^r \mathbb{C}^N) \cong \mathbb{P}^{N_r-1}$ with $N_r = \binom{N}{r}$, and $r$ only affects the target dimension, we may set $r=1$ without loss of generality. This corresponds to $E$ being a line bundle (the case of a non-degenerate band) and to working with maps $f:\mathbb{T}^d\to \mathbb{P}^{N-1}$, with the understanding that $N$ in the target is not the total number of bands if the original Bloch band was degenerate.

\subsection{K\"ahler bands}
\label{subsec: Kaehler bands}

The complex projective space $\mathbb{P}^{N-1}$ carries its canonical Fubini–Study K\"ahler structure, consisting of the symplectic form $\omega_{\mathrm{FS}}$, the complex structure $J$, and the Riemannian metric $g_{\mathrm{FS}} = \omega_{\mathrm{FS}}(\cdot, J\cdot)$.  
For a Bloch band described by a map $f:\mathbb{T}^d \to \mathbb{P}^{N-1}$, the \emph{Berry curvature} and \emph{quantum metric} are precisely the pullbacks
\begin{align}
\omega = f^* \omega_{\mathrm{FS}}, 
\quad
g = f^* g_{\mathrm{FS}}.
\end{align}
Thus, the quantum geometry of the band is precisely the geometry of $\mathbb{P}^{N-1}$ restricted to the image of $f$, justifying the name.

If $f$ is an immersion, the Wirtinger inequality~\cite[4.26, Theorem~4.27]{Ballmann:2006} provides a necessary and sufficient condition for $\mathbb{T}^d$, equipped with $(\omega,g)$, to be an immersed K\"ahler submanifold of $\mathbb{P}^{N-1}$.  
If $v_1,\dots,v_d$ is any basis of the tangent space at a point $[k]\in \mathbb{T}^d$, the inequality reads
\begin{align}
\frac{1}{\left(\tfrac{d}{2}\right)!}\left|\omega^{\wedge \tfrac{d}{2}}(v_1,\dots,v_d)\right|
\ \leq\ 
\mathrm{vol}_g(v_1,\dots,v_d),
\end{align}
with equality if and only if the immersed tangent space is preserved by $J$.  
Saturation everywhere is therefore equivalent to $f$ being holomorphic, which in turn means that $\mathbb{T}^d$ is an immersed K\"ahler submanifold of $\mathbb{P}^{N-1}$.

For $d=2$, the case relevant to this work, and in the coordinates $(k_1,k_2)$ of the universal cover, the inequality takes the simple form
\begin{align}
\left|\omega_{12}\right| \ \leq\ \sqrt{\det(g)}.
\end{align}
When $f$ is an immersion and saturation holds at all points of the Brillouin zone, we say that the corresponding Bloch band is a \emph{K\"ahler band}. Equivalently, a K\"ahler band is a holomorphic immersion $f:\mathbb{T}^2\to\mathbb{P}^{N-1}$.  
If, moreover, the complex structure on $\mathbb{T}^2$ under which $f$ is holomorphic is translation-invariant, we speak of an \emph{ideal K\"ahler band}.

Since every Riemann surface $C$ of genus $g=1$ is isomorphic to a complex torus $\mathbb{C}/\Lambda$ (for some full-rank lattice $\Lambda$), there always exists a diffeomorphism $\phi \in \mathrm{Diff}(\mathbb{T}^2)$ such that $f\circ\phi$ is an ideal K\"ahler band.  
From the physics point of view, composition with $\phi$ can be a drastic transformation, but for classification purposes it is valid; thus it suffices to restrict attention to this case.

Examples of K\"ahler bands include the lowest Landau level, the flat band in the chiral limit of twisted bilayer graphene, (lowest) Landau levels in curved space, and many other constructions~\cite{ozawa:mera:2021,mera:ozawa:2021,mera:ozawa:2021:engineering, wang:cano:millis:liu:yang:2021,wang:liu:2022,ledwith:vishwanath:khalaf:2022,wang:klevtsov:liu:2023,estienne:regnault:valentin:2023,liu:mera:fujimoto:ozawa:2025}.

Crucially, due to Calabi's rigidity theorem~\cite{Cal}, K\"ahler bands—being holomorphic immersions of tori into projective space—are classified up to isometry in projective space by their induced metrics. Concretely, if two K\"ahler bands have the same quantum metric, then they must differ only by an isometry of the target projective space, i.e., by a unitary transformation up to scalar multiples~\cite{mera:ozawa:2024,liu:mera:fujimoto:ozawa:2025}.

It is important to note that the discussion above assumes a real-space unit cell with no \emph{internal spatial structure}. When the unit cell includes multiple orbitals per site whose \emph{spatial centers are distinct}, the Bloch Hamiltonian and ground-state projector become periodic in the quasimomentum $k$ only up to a unitary representation of the reciprocal lattice. In such cases, we are not dealing with maps $f: \mathbb{T}^d \to \mathbb{P}^{N-1}$, but with equivariant maps $f: \mathbb{R}^d \to \mathbb{P}^{N-1}$ satisfying lattice-equivariance conditions. This situation arises, for example, in the Landau level problem—where, in addition, the target projective space is infinite-dimensional. Even in this broader setting, Calabi's rigidity still applies, leading to a strong conclusion: the lowest Landau level is the unique ideal K\"ahler band with flat metric and Chern number one~\cite{wang:klevtsov:liu:2023,mera:ozawa:2024}. The infinite-dimensional nature of the target projective space is precisely what allows for the realization of a perfectly flat quantum metric, a property intimately connected to the unique (up to isomorphism) irreducible representation of the Heisenberg group (the Stone--von Neumann theorem)—corresponding to the magnetic translation symmetry under which Landau level bands are invariant. Such a flat quantum metric is not possible when the target is finite-dimensional. For holomorphic immersions of tori into finite-dimensional projective spaces, the quantum metric is not flat---a fact that can be understood by analyzing the diagonal of the Bergman kernel and its asymptotic behavior~\cite{zelditch:1998,mera:ozawa:2021:engineering}.
\subsection{Harmonic bands}
\label{subsec: harmonic bands}
While K\"ahler bands capture the physics of the lowest Landau level, allowing for non-flat quantum geometry, one would like to have a more general picture that includes higher Landau levels as well. This generalization is captured by the so-called \emph{generalized Landau levels}, or \emph{harmonic bands}~\cite{liu:mera:fujimoto:ozawa:2025,paiva:wang:ozawa:mera:2025,onishi:avdoshkin:fu:2025}. These are bands for which the Dirichlet energy functional~\cite{EW},
\begin{align}
E(f)=\frac{1}{2}\int_{C} ||\dd f||^2 \, \dd\vol_h,
\end{align}
is extremal, and hence they define harmonic maps into projective space. This functional depends only on the Riemann surface structure of the Brillouin zone (here denoted $C$), which is equivalently specified by an orientation and a conformal class of metrics, or a complex structure. The relation between the two pictures is that the metric $h$ is compatible with the complex structure, so that the latter acts as a $90^\circ$ rotation consistent with the positive orientation of the surface.

The Dirichlet energy functional constitutes the leading non-trivial term in the low-energy expansion of the averaged structure factor~\cite{paiva:wang:ozawa:mera:2025,onishi:avdoshkin:fu:2025} of the band—an observable accessible through optical absorption experiments. The averaging procedure accounts for anisotropy in the medium, which is essentially encoded in the choice of a translation-invariant complex structure on the Brillouin zone. Moreover, the Dirichlet energy is essentially the \emph{integrated (generalized) trace of the quantum metric}~\cite{liu:mera:fujimoto:ozawa:2025}, which has recently emerged as a key measure of Landau level physics: for harmonic bands, this integrated trace is quantized in units of the Chern number of the band, i.e., the degree of $f^* \mathcal{O}(1)$.

K\"ahler bands are harmonic maps that minimize $E(f)$. However, there exist non-holomorphic (and thus non-K\"ahler) harmonic bands that still extremize the Dirichlet energy.

Following the construction of Eells and Wood~\cite{EW}, all such bands can be systematically generated as follows. Let $f:C\to \mathbb{P}^{N-1}$ be a K\"ahler band. Since holomorphic line bundles over $\mathbb{C}$ are trivializable, the lift $\widetilde{f}:\mathbb{C} \to \mathbb{C}^N \setminus \{0\}$ can be taken globally (in the genus one case). One then considers the sequence of derivatives $\frac{\dd ^n}{\dd z^n}\widetilde{f}$ for $n=0,1,\dots,N-1$, and applies the Gram--Schmidt orthogonalization procedure. This yields a collection of maps $\widetilde{f}_k:\mathbb{C} \to \mathbb{C}^N$, $k=0,\dots,N-1$, which turn out to descend to well-defined maps $f_k:C \to \mathbb{P}^{N-1}$. Each of these maps is harmonic. Of course, $f_0 = f$, but for $k \neq 0$, the maps are generally not holomorphic, except for the last one, $f_{N-1}$, which is antiholomorphic—i.e., holomorphic with respect to the opposite complex structure.

This geometric construction, deeply tied to K\"ahler geometry, allows to determine the quantum geometry of generalized Landau levels from the quantum geometry of a single starting K\"ahler band~\cite{liu:mera:fujimoto:ozawa:2025}. It also enabled the computation of the Hall viscosity for the many-particle states obtained by filling these bands, which captures how they vary under deformations of the complex structure. Physically, this corresponds to the response of the system to area-preserving deformations of the sample~\cite{paiva:wang:ozawa:mera:2025}.

\section{Mathematical background}
\label{sec: background}

Let $C$ be an elliptic curve, which can be realized as $\mathbb{C} / \Lambda$ for a lattice $\Lambda$ generated by $1$ and an element $\tau$ in the upper half plane corresponding to the complex structure of $C$. Throughout, $\mathbb{P}^{N-1}$ is endowed with the Fubini--Study metric corresponding to the standard Euclidean metric on $\mathbb{C}^{N}$; hereafter, all the $\mathbb{C}$-vector spaces that appear in the following are endowed with the standard Euclidean metrics with respect to the bases specified later. We assume $N \ge 3$.

Eells--Wood \cite[Theorem 6.9]{EW} prove that any harmonic map $C \to \mathbb{P}^{N-1}$ of degree $d >0$ is isotropic, namely that it is obtained by the construction explained in Section \ref{subsec: harmonic bands}. We now present a coordinate-free description of it. Suppose that we have a holomorphic map $f : C \to \mathbb{P}^{N-1}$. For each $k=0, \dots , N-1$ we have the $k$th associated map $f^{(k)} : C \to \mathrm{Gr}(k, \mathbb{P}^{N-1})$, by sending $x \in C$ to the osculating $k$-plane to $f$ at $x$ (see \cite[Chapter I, Exercise C]{GAC1} or \cite[Chapter 2, \S 4]{GH}). Defining a flag variety $\mathcal{F}_k$ as
\begin{equation*}
	\mathcal{F}_k := \{ (A,B) \in \mathrm{Gr}(k-1, \mathbb{P}^{N-1}) \times \mathrm{Gr}(k, \mathbb{P}^{N-1}) \mid A \subset B \} ,
\end{equation*}
we have a smooth (generally non-holomorphic) map
\begin{equation*}
	\pi_k : \mathcal{F}_k \to \mathbb{P}^{N-1}
\end{equation*}
by sending $(A,B) \in \mathcal{F}_k$ to the line $B \cap A^{\perp}$ where $A^{\perp}$ is the orthogonal complement of $A$ in $B$. It turns out that $\pi_k$ is a Riemannian submersion \cite[Section 3, B]{EW}. Eells--Wood \cite[Proposition 3.15]{EW} shows that
\begin{equation} \label{dfeqfkhmmp}
f_k := \pi_k \circ (f^{(k-1)} , f^{(k)}) : C \to \mathbb{P}^{N-1}
\end{equation}
is harmonic. We decree $f_0 :=f$. A harmonic map obtained by this construction is called isotropic or superminimal, and agrees with the one explained in \ref{subsec: harmonic bands}. Since $f_k$ is holomorphic when $k=0$ and antiholomorphic when $k=N-1$ \cite[Remark 3.11]{EW}, we only consider $k=1 , \dots , N-2$ in this paper.

The rigidity of isotropic harmonic maps is an important question, considered in the mathematics literature such as \cite{BO,Chi90,ChiMo,ChiZhe,Cuk}. We formulate this question as below, following Cukierman \cite[(1.2)]{Cuk}.

\begin{question} \label{qrdgish}
	Suppose that two isotropic harmonic maps $f_k, f'_h : C \to \mathbb{P}^{N-1}$, constructed from two holomorphic maps $f, f' : C \to \mathbb{P}^{N-1}$ of the same degree and $k,h=1, \dots , N-2$, are isometric. Does there exist $\sigma \in \mathrm{PU} (N)$ acting linearly on $\mathbb{P}^{N-1}$ such that $f'_h = \sigma \circ f_k$?
\end{question}

This question is motivated by Calabi's rigidity \cite[Theorems 2 and 9]{Cal}, again mentioned in Section \ref{subsec: harmonic bands}, which proves the claim affirmatively for $k=h=0$ (or $N-1$); see also \cite[Example 2]{Gre} or \cite[Appendix]{paiva:wang:ozawa:mera:2025}. Previous works found an important formulation of this question in terms of algebraic geometry, particularly the following due to Cukierman \cite{Cuk} (see also Chi--Mo \cite{ChiMo}). Given a holomorphic map $f : C \to \mathbb{P}^{N-1}$ and $k=1, \dots , N-2$, we define a new holomorphic map $f^{\langle k \rangle}$ by composing $(f^{(k-1)} , f^{(k)})$ with the Pl\"ucker and Segre embeddings as follows,
\begin{displaymath}
			\xymatrixcolsep{5pc}\xymatrixrowsep{2pc}\xymatrix{{C} \ar@{->}[rrdd]_-{f^{\langle k \rangle}} \ar@{->}[r]^-{(f^{(k-1)} , f^{(k)})} & \mathcal{F}_k \ar@{^{(}->}[r]^-{\text{inclusion}}  & \mathrm{Gr}(k-1, \mathbb{P}^{N-1}) \times \mathrm{Gr}(k, \mathbb{P}^{N-1}) \ar@{^{(}->}[d]^-{\text{Pl\"ucker}}  \\
		  & & \mathbb{P} \left( \bigwedge^k \mathbb{C}^{N} \right) \times \mathbb{P} \left( \bigwedge^{k+1} \mathbb{C}^{N} \right) \ar@{^{(}->}[d]^-{\text{Segre}} \\
		  & & \mathbb{P} \left( \bigwedge^{k} \mathbb{C}^{N} \otimes \bigwedge^{k+1} \mathbb{C}^{N} \right) =:\mathbb{P}^{L_{k,N}-1}}
\end{displaymath}
where
\begin{equation*}
	L_{k,N} := \dim \left( \bigwedge^{k} \mathbb{C}^{N} \otimes \bigwedge^{k+1} \mathbb{C}^{N} \right).
\end{equation*}
In the above diagram, the vector spaces underlying the projective spaces and the Grassmannians are endowed with the Euclidean metrics with respect to appropriate bases such that the Pl\"ucker and the Segre embeddings preserve these metrics.

Given two holomorphic maps $f, f' : C \to \mathbb{P}^{N-1}$, Cukierman \cite[(1.7)]{Cuk} proves that $f_k$ is isometric to $f'_h$ if and only if $f^{\langle k \rangle}$ is isometric to $f'^{\langle h \rangle}$. Thus, by Calabi's rigidity, there exists $\tilde{\sigma} \in \mathrm{PU}(N)$ such that $f'^{\langle k \rangle} = \tilde{\sigma} \circ f^{\langle h \rangle}$. Hence answering Question \ref{qrdgish} amounts to proving that $\tilde{\sigma}$ comes from the unitary linear action on $\mathbb{P}^{N-1}$, by \cite[Proposition 1.11]{Cuk}.

From the definition (\ref{dfeqfkhmmp}) of $f_k$ and the diagram above, we find that the pullback bundle $f_k^* \mathcal{O}_{\mathbb{P}^{N-1}} (-1)$ is smoothly isomorphic to the quotient $C^{\infty}$-complex vector bundle $(f^{(k)})^*U_k / (f^{(k-1)})^*U_{k-1}$ of rank one, where $U_k$ (resp.~$U_{k-1}$) is the universal bundle over $\mathrm{Gr}(k, \mathbb{P}^{N-1})$ (resp.~$\mathrm{Gr}(k-1, \mathbb{P}^{N-1})$). This isomorphism preserves the metrics, since the linear map underlying $\pi_k$ does.

\section{Results}
\label{sec: result}
While previous works \cite{BO,Chi90,ChiMo,ChiZhe,Cuk} consider Question \ref{qrdgish} for general Riemann surfaces, we prove in this paper that a stronger result is available when $C$ is an elliptic curve, up to a contribution from the automorphism group. Recall first that, for any pair of non-degenerate holomorphic embeddings $f, f' : C \to \mathbb{P}^{N-1}$ of degree $N$, there exists $\sigma \in \mathrm{PGL}(N , \mathbb{C})$ and a holomorphic automorphism $\alpha$ of $C$ such that $f'=\sigma \circ f \circ \alpha$ (see Lemma \ref{lmpreqell}), where we can take $\alpha$ to be the identity if $f, f'$ are associated to the same complete linear system (see Remark \ref{rmrrddm}). In the following theorem, we assume that $\alpha$ preserves the quantum metric.

\begin{theorem} \label{thmrgd}
	Let $C$ be an elliptic curve, $k,h=1, \dots , N-2$, and let $f_k, f'_h : C \to \mathbb{P}^{N-1}$ be two isotropic harmonic maps of degree $d>0$ constructed from non-degenerate holomorphic embeddings $f, f' : C \to \mathbb{P}^{N-1}$ of degree $N$, such that $f'=\sigma \circ f \circ \alpha$ for some $\sigma \in \mathrm{PGL}(N , \mathbb{C})$ and a holomorphic isometry $\alpha$ of $f_{h}'^* g_{\mathrm{FS}}$.
	
	Suppose that $f_k$ is isometric to $f'_h$. Then $k=h$, and there exists $\hat{\sigma} \in \mathrm{PU} (N)$ acting linearly on $\mathbb{P}^{N-1}$ such that $f'_h = \hat{\sigma} \circ f_k \circ \alpha$. Moreover, $f'_l = \hat{\sigma} \circ f_l \circ \alpha$ holds for all $l=1, \dots , N-2$, which are isometric if $\alpha$ is a holomorphic isometry of $f_{l}'^* g_{\mathrm{FS}}$.
\end{theorem}

We recall that a holomorphic map $f : C \to \mathbb{P}^{N-1}$ is said to be non-degenerate (also called full) if the image $f(C)$ is not contained in any proper linear subspace of $\mathbb{P}^{N-1}$.

\begin{remark} \label{rmrrddm}
	The hypothesis that $f$ and $f'$ have degree $N$ is equivalent to these maps being associated to complete linear systems on $C$. Indeed, the Riemann--Roch theorem implies
	\begin{equation*}
		h^0(C,L) = h^0(C, K_C \otimes L^{-1}) + \int_C c_1(L),
	\end{equation*}
	where $L:= f^* \mathcal{O}_{\mathbb{P}^1} (1)$ is very ample since $f$ is an embedding. Since the canonical bundle $K_C$ of $C$ is trivial and $\int_C c_1(L) >0$, we have $\int_C c_1(L) = h^0(C,L) \ge N$ since $f$ is non-degenerate, with equality if and only if $f$ is associated to a complete linear system $f: C \to \mathbb{P}(H^0(C, L)^{\vee})$ (see also \cite[page 253]{GH} or \cite[(2.1)]{Cuk}). Throughout this paper, until we mention otherwise, we assume that the degrees of $f, f'$ are both $N$, which is natural since concrete examples of projective embeddings are primarily constructed out of complete linear systems.
\end{remark}

We closely follow the approach of Cukierman \cite{Cuk}. We first prove $k=h$, and that the degree $d$ of the associated isotropic maps equals $N$ (see Lemma \ref{lmhpmkh}). Given $f : C \to \mathbb{P}^{N-1}$ and $k=1, \dots , N-2$, recall that we have a divisor $R_k$ on $C$ defined by
\begin{equation*}
	R_k := \sum_{x \in C} \sum_{0 \le j \le k} \alpha_j (x) x
\end{equation*}
where $\alpha_j (x)$ is the $j$th ramification index at $x$ of the linear series associated to $f$ (see \cite[\S 2]{Cuk} and \cite[Chapter I, Exercise C]{GAC1} for more details). Following \cite[(2.10)]{Cuk}, we say that $f$ is without special hyperosculation points if $R_{N-2}=0$.

\begin{lemma} \label{lmhposc}
	Let $f : C \to \mathbb{P}^{N-1}$ be a non-degenerate holomorphic embedding from an elliptic curve of degree $N$, and let $f^{(k)}$ be the $k$th associated map for $k=1, \dots , N-2$. Then the $j$th ramification index of $f$ is zero at every $x \in C$ for any $j=0, \dots , k$. In particular, $f$ has no special hyperosculation points, and the $k$th associated map $f^{(k)}$ is an immersion at every $x \in C$ for any $k=1, \dots , N-2$.
\end{lemma}

\begin{proof}
	We first note that the holomorphic line bundle $f^* \mathcal{O}_{\mathbb{P}^{N-1}} (1)$ on $C$ is of degree $N \ge 3$ by Riemann--Roch, as in Remark \ref{rmrrddm}. The first claim of the lemma is again a consequence of Riemann--Roch, which ensures the existence of a holomorphic section $s$ of $f^* \mathcal{O}_{\mathbb{P}^{N-1}} (1)$ with the following properties: for any given $x \in C$ and any given $m=0, \dots , N-2$, there exists $s$ which vanishes at $x$ with order $m$ (see e.g.~\cite[\S 2, Example 2]{Pie}). It follows immediately that there are no special hyperosculation points of $f$, i.e.~$R_{N-2}=0$, and also that $f^{(k)}$ is an immersion by \cite[Chapter I, Exercise C-6]{GAC1}.
\end{proof}

\begin{lemma} \label{lmhpmkh}
	Let $f , f' : C \to \mathbb{P}^{N-1}$ be non-degenerate holomorphic embedding of an elliptic curve of degree $N$, and $k,h=1, \dots , N-2$, such that the associated isotropic harmonic maps $f_k$, $f'_h$ are isometric. Then $k=h$, and their degrees are both equal to $N$.
\end{lemma}

\begin{proof}
	Lemma \ref{lmhposc} shows that both $f$ and $f'$ have no special hyperosculation points. Thus \cite[Proposition 2.9]{Cuk}, together with Calabi's rigidity applied to $f^{\langle k \rangle }$ and $f'^{\langle h \rangle }$, implies $h=k$ as in \cite[Corollary 2.11]{Cuk}; note that, when $L_{k,N} < L_{h,N}$, we have an isometric embedding $\mathbb{P}^{L_{k,N}-1} \hookrightarrow \mathbb{P}^{L_{h,N}-1}$ to which Calabi's rigidity can be applied (see \cite[Example 2]{Gre} or \cite[Appendix]{paiva:wang:ozawa:mera:2025}). The degrees of $f_k$, $f'_h$ are both equal to $N$ by \cite[Proposition 7.1 (iii)]{EW}.
\end{proof}

The following lemma, which should be well-known to experts, is just to motivate the hypotheses of Theorem \ref{thmrgd} and is not a logical part of the proof.

\begin{lemma} \label{lmpreqell}
	For any two non-degenerate holomorphic embeddings $f, f' : C \to \mathbb{P}^{N-1}$ of an elliptic curve $C$ of degree $N$, there exists $\sigma \in \mathrm{PGL}(N , \mathbb{C})$ and a holomorphic automorphism $\alpha$ such that $f'=\sigma \circ f \circ \alpha$.
\end{lemma}

\begin{proof}
	We provide two proofs. We recall from Remark \ref{rmrrddm} that the embedding $f$ in the lemma is associated to a complete linear series $f: C \to \mathbb{P}(H^0(C, L)^{\vee})$ where $L:= f^* \mathcal{O}_{\mathbb{P}^{N-1}} (1)$. Similarly, $f'$ can be written as $f': C \to \mathbb{P}(H^0(C, L')^{\vee})$. Since $\int_C c_1(L) = N = \int_C c_1(L')$, we find that there exists $F \in \mathrm{Pic}^0(C)$ such that $L'=L \otimes F$, where $\mathrm{Pic}^0(C)$ is the Picard variety of $C$ which classifies all degree $0$ holomorphic line bundles on $C$. Now, the Jacobi inversion theorem (see \cite[Chapter 2, \S 2]{GH} or \cite[Chapter 1, \S 3]{GAC1}) shows that $C$ is isomorphic to $\mathrm{Pic}^0(C)$ as an Abelian variety. Thus, writing $\alpha \in \mathrm{Aut}(C)$ for the translation in $C$ which moves $F$ to $\mathcal{O}_C$ under the isomorphism $C \xrightarrow{\sim} \mathrm{Pic}^0(C)$, we find $L'= \alpha^* L$.
	
	Another proof can be obtained by the dimension count of the Hilbert scheme. Recall \cite[Proposition 5.18, Chapter IX]{GAC2} that the dimension of the Hilbert scheme of degree $N$ curves in $\mathbb{P}^{N-1}$ at $C$ is given by $N^2 = \dim \mathrm{PGL}(N , \mathbb{C}) + 1$. Noting that the moduli space of elliptic curve is 1-dimensional (i.e.~the $j$-invariant, or the choice of the complex structure $\tau$), we find that any holomorphic embedding of $C$ in $\mathbb{P}^{N-1}$ is $\mathrm{PGL}(N , \mathbb{C})$-equivalent to the one given by the linear system $f: C \to \mathbb{P}(H^0(C, L)^{\vee})$, up to holomorphic automorphisms of $C$. Indeed, the existence of a projectively non-equivalent non-degenerate embedding $\tilde{f}$ would yield a family of holomorphic embeddings $f + t \tilde{f}$, of $C$ (with the fixed complex structure), for any small enough $t \in \mathbb{C}$. Differentiating with respect to $t$ at $t=0$, we find that $\tilde{f}$ is a tangent vector of the Hilbert scheme at $f$. The dimension count as above implies that there exists $\xi \in \mathfrak{pgl}(N, \mathbb{C}) = \mathfrak{sl}(N, \mathbb{C}) $ such that $\tilde{f}$ and $\xi f$ (in the sense of a matrix acting on a vector) represent the same point in the projective space for each point in $C$, again up to holomorphic automorphisms of $C$. If $\xi$ is invertible, the conclusion above contradicts that $\tilde{f}$ is not projectively equivalent to $f$. If $\xi$ is not invertible, it contradicts the assumption that $\tilde{f}$ is non-degenerate.
\end{proof}

\begin{remark}
	A more direct proof is available for $N=3$, which also covers the case when $C$ is singular. Indeed, as explained in \cite[\S 10.3]{Dol}, each of
	\begin{itemize}
		\item non-singular cubic curves
		\item irreducible nodal curves
	\end{itemize}
	forms a single $\mathrm{PGL}(3,\mathbb{C})$-orbit in $\mathbb{P}^2$.
\end{remark}

Lemma \ref{lmpreqell} does not hold when the degree $d'$ of the embedding is strictly larger than $N$, which happens when it is associated to an incomplete linear system on $C$ (see Remark \ref{rmrrddm}). Indeed, the dimension of the corresponding Hilbert scheme is given \cite[Proposition 5.18, Chapter IX]{GAC2} by $Nd' > N^2 = \dim \mathrm{PGL}(N , \mathbb{C}) + 1$, and hence the deformation of complex structures of $C$ and the $\mathrm{PGL}(N , \mathbb{C})$-action cannot exhaust all projective embeddings of $C$ in $\mathbb{P}^{N-1}$, meaning that there has to be a projectively non-equivalent embedding of $C$ in $\mathbb{P}^{N-1}$.
	
	Similarly, Lemma \ref{lmpreqell} does not hold when the genus $g$ of $C$ is greater than 1 and $N > g-1$, even when the embedding is associated to a complete linear system. The dimension of the corresponding Hilbert scheme is again given \cite[Proposition 5.18, Chapter IX]{GAC2} by
\begin{equation*}
	Nd'-(N-4)(g-1),
\end{equation*}
where we have
\begin{equation*}
	d'=N+(g-1)
\end{equation*}
by Riemann--Roch, when $2g-2<d' \iff N >g-1$. We find
\begin{equation*}
	Nd'-(N-4)(g-1) = N^2+(3g-3)+(g-1) > \dim \mathrm{PGL}(N , \mathbb{C}) + \dim \mathcal{M}_{g}
\end{equation*}
where $\mathcal{M}_g$ is the moduli space of genus $g$ curves. Thus the deformation of complex structures of $C$ and the $\mathrm{PGL}(N , \mathbb{C})$-action cannot exhaust all projective embeddings of $C$ in $\mathbb{P}^{N-1}$. See also \cite[Theorem 2]{BO} and \cite[Remark 2.12]{Cuk} for the case when $g=0$.

\begin{proof}[Proof of Theorem \ref{thmrgd}]
	As $k=h$ follows from Lemma \ref{lmhpmkh}, it suffices to prove the latter claim. We have $f'=\sigma \circ f \circ \alpha$ for some $\sigma \in \mathrm{PGL}(N , \mathbb{C})$ with $\alpha$ being a holomorphic isometry of $f_{h}'^* g_{\mathrm{FS}}$, by the hypothesis. We prove that $\sigma$ is unitary (up to scalar multiples); indeed, recalling that $f_k$ is constructed as a line in the $k$-plane generated by the jets of $f$ that is orthogonal to the $(k-1)$-plane, a projective unitary equivalence $\sigma \in \mathrm{PU}(N)$ between $f \circ \alpha$ and $f'$ gives the required relationship $f'_h = f'_k =\sigma \circ f_k \circ \alpha$, and moreover $f'_l = \sigma \circ f_l \circ \alpha$ for all $l=1, \dots , N-2$.

	Recall that the hypothesis of the theorem implies that $f^{\langle k \rangle}$ is isometric to $f'^{\langle k \rangle}$ by \cite[(1.7)]{Cuk}. Moreover, $f'^{\langle k \rangle} \circ \alpha^{-1}$ is isometric to $f'^{\langle k \rangle}$ (and hence to $f^{\langle k \rangle}$) since $\alpha$ is a holomorphic isometry of $f_{k}'^* g_{\mathrm{FS}} = f_{h}'^* g_{\mathrm{FS}}$. Thus, by Calabi's rigidity there exists $\tilde{\sigma} \in \mathrm{PU}(L_{k,N})$ such that $f'^{\langle k \rangle} \circ \alpha^{-1}= \tilde{\sigma} \circ f^{\langle k \rangle}$. Let $\rho (\sigma)$ be the action of $\sigma$ naturally induced on $\bigwedge^{k} \mathbb{C}^{N} \otimes \bigwedge^{k+1} \mathbb{C}^{N}$. We have $f'^{\langle k \rangle} = \rho (\sigma) \circ f^{\langle k \rangle} \circ \alpha$, and hence $\tilde{\sigma}^{-1} \rho (\sigma)$ satisfies
	\begin{equation} \label{eqexeqstr}
		( \tilde{\sigma}^{-1} \rho (\sigma) ) \circ f^{\langle k \rangle} (x) = f^{\langle k \rangle} (x)
	\end{equation}
	for all $x \in C$. Since $f$ has no special hyperosculation points by Lemma \ref{lmhposc}, $f^{\langle k \rangle} : C \to \mathbb{P}^{L_{k,N}-1}$ is a holomorphic map associated to a linear series of dimension $L_{k,N}$ on $C$ as explained in \cite[(2.8)]{Cuk}, which implies that $f^{\langle k \rangle} (C)$ is not contained in any proper linear subspace of $\mathbb{P}^{L_{k,N}-1}$; this fact also follows from \cite[(4.1)]{Cuk}, by noting that the kernel of the higher Gauss map is trivial by Lemma \ref{lmhposc} and the Pl\"ucker and Segre embeddings are non-degenerate. Thus, the equation (\ref{eqexeqstr}) forces $\tilde{\sigma}^{-1} \rho (\sigma)$ to be a constant multiple of the identity matrix. Hence $\rho ( \sigma )$ can be represented by a projective unitary matrix, and \cite[Proposition 1.11]{Cuk} immediately implies that $\sigma$ itself is unitary (up to scalar multiples).
\end{proof}
We can drop the assumptions on the degree of $f,f'$ and the automorphism $\alpha$, provided we make the additional assumptions that they have no special hyperosculation points and also that the pullbacks of the Fubini--Study symplectic form $\omega_{\mathrm{FS}}$ on $\mathbb{P}^{N-1}$ coincide. 
\begin{theorem}\label{th: rigidity of harmonic maps under no hyperosculation pts}
Let $C$ be an elliptic curve, $k,h=1, \dots , N-2$, and let $f_k, f'_h : C \to \mathbb{P}^{N-1}$ be two isotropic harmonic maps of degree $d>0$ constructed from non-degenerate holomorphic embeddings $f, f' : C \to \mathbb{P}^{N-1}$ without special hyperosculation points.
Suppose that $f_k$ is isometric to $f'_h$ and that $f^*_k\omega_{\mathrm{FS}}=f_{h}'^*\omega_{\mathrm{FS}}$. Then $k=h$, and there exists $\sigma \in \mathrm{PU}(N)$ acting linearly on $\mathbb{P}^{N-1}$ such that $f'_h = \sigma \circ f_k$, and moreover $f'_l = \sigma \circ f_l$ holds for all $l=1, \dots , N-2$.    
\end{theorem}
\begin{proof}
We first note that $k=h$. Since $f_k$ being isometric to $f'_h$ is equivalent to $f^{\langle k\rangle}$ being isometric to $f^{'\langle h\rangle}$,~\cite[(1.7)]{Cuk} , it follows from the assumption in the theorem that the latter two maps are isometric. Calabi's rigidity then implies that the two maps define isomorphic linear series and hence applying~\cite[Proposition (2.9)]{Cuk} with the assumption that there are no special hyperosculation points, we get the result that $k=h$.

Next, observe that if we prove that $f$ and $f'$ determine the same K\"ahler metric on $C$, it follows by Calabi's rigidity theorem that $f'=\sigma \circ f$, with $\sigma \in PU(N)$. Because $f_k$ and $f'_h=f'_k$ are, respectively, the lines in the planes generated by the $k$-jets of $f$ and $f'$ which are orthogonal to the $(k-1)$-planes generated by the $(k-1)$-jets of $f$ and $f'$, it follows that such $\sigma$ will yield the required projective unitary. 

Writing $\omega_{\mathrm{FS}, \langle j\rangle}$ for the Fubini--Study symplectic form on $\mathbb{P}^{L_{j,N}-1}$, we denote by $\omega^{\langle j\rangle}=(f^{\langle j\rangle} )^{*}\omega_{\mathrm{FS}, \langle j\rangle}$ and $\omega_{j}=f_j^{*}\omega_{\mathrm{FS}}$, for $j=0,\dots N-2$, and similarly for the primed maps. The assumptions in the theorem, together with the discussion in the beginning of the proof implies
\begin{align}
\omega^{\langle k\rangle} &=(f^{\langle k\rangle} )^{*}\omega_{\mathrm{FS}, \langle k \rangle}=(f^{'\langle k\rangle} )^{*}\omega_{\mathrm{FS}, \langle k \rangle}=\omega^{'\langle k\rangle}, \nonumber\\
\omega_{k} &=f_{k}^{*}\omega_{\mathrm{FS}}=f_{k}'^{*}\omega_{\mathrm{FS}}=\omega'_{k}
\end{align}
and, furthermore, denoting $g_k=\omega^{\langle k\rangle}(\cdot, j\cdot)$, for $j$ the complex structure of $C$, we have
\begin{align}
g_k=f_k^*g_{\mathrm{FS}}=(f'_k)^*g_{FS}.
\end{align}
Using the relation between the maps $f^{\langle j\rangle}$ and $f_j$ and the maps $f^{(j)}$ and $f^{(j-1)}$, we have, see Ref.~\cite[Appendix~B5 and~B6]{liu:mera:fujimoto:ozawa:2025},
\begin{align}
\omega^{\langle j\rangle}=\omega^{(j)}+\omega^{(j-1)} \text{ and } \omega_{j}=\omega^{(j)}-\omega^{(j-1)},
\end{align}
and similarly for the primed quantities, with the understanding that $\omega^{(-1)}=\omega^{'(-1)}=0$ (because $\bigwedge^0\mathbb{C}^{N}=\mathbb{C}$). Note that the above equations also follows from the diagram in Section \ref{sec: background} and the metric preserving $C^{\infty}$-isomorphism $f_j^* \mathcal{O}_{\mathbb{P}^{N-1}} (-1) \cong (f^{(j)})^*U_j / (f^{(j-1)})^*U_{j-1}$ discussed therein, since $\omega^{(j)}-\omega^{(j-1)}$ is the curvature form of the metric connection on the dual of this quotient bundle. Observe that, under the hypothesis of the theorem, we can extract $\omega^{(k-1)}$ and $\omega^{(k)}$ and they coincide with $\omega^{'(k-1)}$ and $\omega^{'(k)}$, respectively.

We then appeal to the second main theorem of holomorphic curves~\cite[Eq.~(4.10)]{griffiths:1974}, which in our notation reads
\begin{align}
\Ric{\omega^{(j)}}=\omega^{(j+1)}-2\omega^{(j)}+\omega^{(j-1)},\ j=0,\dots,N-2,
\end{align}
for the map $f$. Similar equation holds for the map $f'$ with primed quantities. Observe that the equations above are well-defined as equations amongst smooth $2$-forms due to the assumption that $f$ and $f'$ have no special hyperosculation points. That assumption implies that the induced maps $f^{(j)}$ and $f^{'(j)}$, $j=0,\dots, N-2$ are holomorphic immersions themselves by Lemma \ref{lmhposc}, meaning that $\omega^{(j)}$ and $\omega^{'(j)}$ are \emph{bona fide} symplectic forms, $j=0,\dots, N-2$.  In the above equations, we also take, for $j=N-1$, $\omega^{(N-1)}=\omega^{'(N-1)}=0$ (because $\bigwedge^N\mathbb{C}^{N}=\mathbb{C}$). Observe that the above are recurrence equation of order $2$, which need two initial conditions to be solved and there is no guarantee that the solution exists or is unique. In the setup of~\cite{griffiths:1974}, the initial conditions are $\omega^{(-1)}$ and $\omega^{(0)}$ and the solution exists and is unique. Indeed, in that case we have $\omega^{(-1)}=0$ (convention) and $\omega^{(0)}=f^*\omega_{\mathrm{FS}}$, this gives, for $j=0$,
\begin{align}
\Ric{\omega^{(0)}}=\omega^{(1)}-2\omega^{(0)},
\end{align}
from which we can extract $\omega^{(1)}$. For $j=1$, we have
\begin{align}
\Ric{\omega^{(1)}}=\omega^{(2)}-2\omega^{(1)}+\omega^{(0)},
\end{align}
from which we can extract $\omega^{(2)}$. We can proceed by induction, and the process terminates at $j=N-2$, with the equation
\begin{align}
\Ric{\omega^{(N-2)}}=-2\omega^{(N-2)}+\omega^{(N-3)}.
\end{align}
This means that, in the case $f$ has no special hyperosculation points, from the knowledge of $\omega^{(0)}$, we can extract all the other $\omega^{(j)}$'s uniquely. The same occurs for $f'$ and the $\omega^{'(j)}$'s. In the hypothesis of our theorem, as discussed above, we are given $\omega^{(k-1)}=\omega'^{(k-1)}$ and $\omega^{(k)}=\omega'^{(k)}$. We now show that with this information, the recurrence relations of the second main fundamental theorem of holomorphic curves can be solved and the solution is also unique. To see this, observe that the equation for $j=k$ yields
\begin{align}
\Ric{\omega^{(k)}}=\omega^{(k+1)}-2\omega^{(k)}+\omega^{(k-1)},
\end{align}
from which we can extract $\omega^{(k+1)}$, and the equation for $j=k-1$ yields
\begin{align}
\Ric{\omega^{(k-1)}}=\omega^{(k)}-2\omega^{(k-1)}+\omega^{(k-2)},
\end{align}
from which we can extract $\omega^{(k-2)}$. We can now go up in the recursion relation because for the $j=(k+1)$th step we have the data of $\omega^{(k)}$ and $\omega^{(k+1)}$ and we can also go down in the recursion relation because for the $j=(k-2)$th step we have $\omega^{(k-2)}$ and $\omega^{(k-1)}$. This then ensures that the induction process terminates, giving a unique solution for the sequence of $2$-forms $\omega^{(j)}$. Because the initial conditions are the same and the solution is unique, the resulting sequence of $2$-forms is also the same and, in particular, it follows that $\omega^{(0)}=\omega^{'(0)}$ which then implies that $f$ and $f'$ determine the same K\"ahler metric (since $f^{(0)}=f$ and $f^{'(0)}=f'$) and hence finishes the proof of the theorem. 
\end{proof}

\section*{Acknowledgments}
\begin{sloppypar}
YH is supported by JSPS KAKENHI Grant Number JP23K03120 and JP24K00524, and thanks Professor Yoshihiro Ohnita for helpful discussions. B.~M. acknowledges support from the Security and Quantum Information Group (SQIG) in Instituto de Telecomunica\c{c}\~{o}es, Lisbon. This work is funded by FCT/MECI through national funds and when applicable co-funded EU funds under UID/50008: Instituto de Telecomunicações (IT). B.~M. further acknowledges the Scientific Employment Stimulus --- Individual Call (CEEC Individual) --- 2022.05522.CEECIND/CP1716/CT0001, with DOI: \href{https://doi.org/10.54499/2022.05522.CEECIND/CP1716/CT0001}{10.54499/2022.05522.CEECIND/CP1716/CT0001}. TO is supported by JSPS KAKENHI Grant Number JP24K00548 and JST PRESTO Grant No. JPMJPR2353.
\end{sloppypar}
\bibliographystyle{JHEP.bst} 
\bibliography{harmonic.bib}
\end{document}